\documentclass[prd,amsmath,amssymb,showpacs]{revtex4}

\usepackage{pstricks}
\usepackage{amsthm}
\usepackage{mathrsfs,latexsym}
\usepackage{bm}

\newcommand{\real}{\mathbb{R}}
\newcommand{\complex}{\mathbb{C}}

\newtheorem{Thm}{Theorem}[section]

\newtheorem{Prop}[Thm]{Proposition}

\numberwithin{equation}{section}

\newcommand{\A}{\mathcal{A}}
\newcommand{\B}{\mathcal{B}}
\newcommand{\F}{\mathcal{F}}
\newcommand{\J}{\mathcal{J}}
\newcommand{\U}{\mathcal{U}}
\newcommand{\V}{\mathcal{V}}

\newcommand{\supp}{{\rm supp}\,}

\newcommand{\bd}{{\bf d}}            
\newcommand{\bdelta}{\bm{\delta}}    

\newcommand{\stM}{{\bm M}}           

\newcommand{\id}{\openone}

\begin{document}


\title{Quantization of the Maxwell field in curved spacetimes of arbitrary
dimension\\
(corrected version)}
\author{Michael J. Pfenning}
\email{Michael.Pfenning@usma.edu}
\affiliation{Department of Physics, United States Military Academy,%
 West Point, New York, 10996-1790, USA}

\date{July 31, 2013}

\begin{abstract}
We quantize the massless $p$-form field that obeys the generalized Maxwell
field equations in curved spacetimes of dimension $n\ge2$. We begin by
showing that the classical Cauchy problem of the generalized Maxwell
field is well posed and that the field possess the expected gauge
invariance.  Then the classical phase space is developed in terms of
gauge equivalent classes, first in terms of the Cauchy data and then
reformulated  in terms of Maxwell solutions.  The latter is employed
to quantize the field in the framework of Dimock.  Finally, the resulting
algebra of observables is shown to satisfy the wave equation with
the usual canonical commutation relations.
\end{abstract}

\pacs{04.62.+v, 03.50.-z, 03.70.+k, 11.15.Kc, 11.25.Hf, 14.80.-j}
\maketitle

\section{Introduction}

Tensor and/or spinor fields of assorted types in curved spacetime have been
studied by various authors, Buchdahl \cite{Buch58,Buch62,Buch82a,Buch82b,Buch84,%
Buch87}, Gibbons \cite{Gibb76}, and Higuchi \cite{Higu89} to name a few.
It is often found that the straightforward generalization of the flat
spacetime field equations to an arbitrary curved spacetime is wrought with
internal inconsistencies unless specific conditions are met.  For example, in
his earliest work Buchdahl \cite{Buch62} shows that ``the natural field
equations for particles of spin $\frac{3}{2}$ are consistent if and only if
the (pseudo)Riemann space(time) in which they are contemplated is an Einstein
space(time),'' that is one for which the Ricci tensor satisfies $R_{ab} =
\lambda g_{ab}$. Gibbons demonstrates that there is also a breakdown in the
Rarita-Schwinger formulation of spin $\frac{3}{2}$ in four-dimensional
curved spacetime. Meanwhile, Higuchi derives the constraint condition,
$
\nabla_c R_{ab} = \left[\frac{1}{18}(g_{cb}\partial_a + g_{ca}\partial_b)
+\frac{2}{9}g_{ab}\partial_c\right]R,
$
on the background metric for the generalization of the massive symmetric tensor
field equations on curved spacetime.  There are two cases where the above condition
is met: the metric is a solution to the vacuum Einstein equations or the Ricci
tensor $R_{ab}$ is covariantly constant.

In this paper, we study the $p$-form field theory (fully antisymmetric rank-$p$
tensors) in curved spacetimes of arbitrary dimension which obey the
generalization of the Maxwell equations. Unlike the examples above, the
$p$-form field equations generalize to any dimension without inconsistencies
or the need for constraints. Thus, $p$-form field theories seem to be a
natural model for the study of quantum field theories in higher dimension
curved spacetime.  For example, the minimally coupled scalar field ($p=0$
in any dimension) and the electromagnetic field ($p=1$ in four dimensions)
are two examples of this self-consistent $p$-form theory.  Better still,
the quantization of the classical theory turns out to be identical for all $p$-%
form fields, independent of the rank of form.

This manuscript proceeds in the following order. In Section~%
\ref{sec:Gen_Maxwell_field} we discuss the classical generalized Maxwell
field. This begins with a review of electrodynamics in four dimensions
and then proceeds to the generalization of the Maxwell equations into
arbitrary dimension. It is at this point that we convert from the
traditional notation used in relativistic physics to that of exterior
differential calculus thus giving us $p$-form fields.  We then discuss
fundamental solutions to the resulting wave equation and the initial
value problem for the classical field. We shall see that there exists
a gauge freedom in the field which complicates the uniqueness of solutions
for a given set of initial data.  Thus, the Cauchy problem is only well
posed if we work with gauge equivalent classes.

In Section~\ref{sec:Quantization} we quantize the $p$-form field.  From
the classical field theory, we obtain a symplectic phase space consisting
of a real vector space and a symplectic form.  This phase space is
quantized by promoting functions in the phase space to operators acting
on a Hilbert space while simultaneously requiring the commutator of
such operators to be $-i$ times their classical Poisson bracket.  In
this way, the algebra of observables for the quantized field on a manifold
is obtained. Finally, there will be discussion and conclusions.


Throughout this paper we will use units where $\hbar = c =1$. The notation
$C_0^\infty(\real^n)$ denotes the space of smooth, compactly supported%
~\footnote{The {\em support} of a function is the closure of the set of
points on which it is nonzero.}, complex-valued functions on $\real^n$.
We take $\stM$ to be a smooth $n$-dimensional manifold (without boundary)
which is connected, orientable, Hausdorff, paracompact, and equipped with a
smooth metric of index $s$~\footnote{The index is the number of spacelike
(i.e., negative norm-squared) basis vectors in any $g$-orthonormal frame.}.
We denote the space of smooth, complex-valued $p$-forms on $\stM$ by
$\Omega^p(\stM)$; the subspace of compactly supported $p$-forms will be
written $\Omega_0^p(\stM)$. Each $p$-form may be regarded as a fully
antisymmetric covariant $p$-tensor field. Our conventions for forms are
consistent with that of Abraham, Marsden and  Ratiu (AMR)~\cite{AMR88}
and are also summarized in our earlier paper \cite{Pfen03} to which we
refer the reader. Therefore, we only introduce the remaining notational
necessities here to make this paper sufficiently self contained.

The exterior product between forms will be denoted by $\wedge$, the
exterior derivative on forms will be denoted by $\bd$, the Hodge
$*$--operator by $*$ and the co-derivative by $\bdelta$.  The Laplace-%
Beltrami operator, simply called the Laplacian on a Riemannian manifold,
is defined $\Box =-\left(\bdelta\bd + \bd\bdelta\right)$. All these are
consistent with the previous paper.

The only difference between the preceding paper and present one is the
definition for the symmetric pairing $\langle\cdot,\cdot\rangle$ of
$p$-forms under integration;
\begin{equation}
\langle \U ,\V \rangle_\stM \equiv \int_\stM \U \wedge *\V
\end{equation}
for any $\U,\V\in\Omega^p(\stM)$ for which the integral exists. Also,
for smooth $\U\in\Omega^{p-1}(\stM)$ and $\V\in\Omega^p(\stM)$ we have
$\bd(\U\wedge*\V) = \bd\U\wedge* \V -\U\wedge*\bdelta\V$, therefore by
Stokes' theorem
\begin{equation}
\langle \bd\U,\V\rangle_\stM = \langle \U,\bdelta\V\rangle_\stM
\label{eq:d_delta_duality}
\end{equation}
whenever the supports of the forms have compact intersection.  In this
sense the operators $\bd$ and $\bdelta$ are dual.


\section{Classical analysis of the generalized Maxwell field}
\label{sec:Gen_Maxwell_field}

\subsection{Classical electrodynamics in four dimensions}
In Minkowski spacetime it is most common to study the electromagnetic field in
the abstract index notation \cite{Wald_GR} where $F_{ab}$ is the covariant
field--strength tensor and the Maxwell equations are
\begin{equation}
\partial_{[a}F_{bc]} = 0\qquad\mbox{and}\qquad
\partial^a F_{ab} = - 4\pi j_b.
\label{eq:flat_space_maxwell}
\end{equation}
Here $j_b$ is the current density, $\partial_a$ is the partial derivative,
lowering or raising of indices is done with respect to the metric $\eta_{\mu\nu}
=\mbox{diag}(1,-1,-1,-1)$ and its inverse respectively, and $[\ ]$ in the
homogeneous Maxwell equation is shorthand for the antisymmetric permutation
over the indices.

It is known that the generalization of the Maxwell equations to curved
four--dimensional spacetimes is internally consistent. This is accomplished
by the minimal substitution rule; replace the partial derivatives with
$\nabla_a$, the unique covariant derivative operator associated with the
spacetime metric such that $\nabla_a g_{bc} = 0$. The Maxwell equations then
become
\begin{equation}
\nabla_{[a}F_{bc]} = 0\qquad\mbox{and}\qquad
\nabla^a F_{ab} = - 4\pi j_b.
\label{eq:conventional_maxwell}
\end{equation}
It is also common to introduce the (co)vector potential $A_a$ related
(at least locally) to the field strength tensor by $F_{ab} = \nabla_{[a} A_{b]}$.
Recast in $A_a$, the homogeneous equation is trivially satisfied, as a result
of the first Bianchi identity, and the inhomogeneous equation becomes
\begin{equation}
\nabla^{a}\nabla_{a} A_b -\nabla_{b}\nabla^{a} A_a - {R^{a}}_b A_a = - 4\pi j_b.
\end{equation}
The Maxwell equations, in both flat and curved spacetime have a gauge
freedom in that many different forms of $A_a$ give rise to the same
$F_{ab}$.  This comes from the freedom to add to $A_a$ the gradient
of any scalar function $\Lambda$. Because the covariant derivatives,
like partial derivatives, commute when acting on scalars the addition
of the gradient term has no effect on the final outcome of the resulting
field strength.  Therefore one can choose to work in a particular gauge,
say the Lorenz gauge where $\nabla^{a} A_a = 0$. Then we have considerable
simplification to the globally hyperbolic equation
\begin{equation}
\nabla^{a}\nabla_{a} A_b - {R^{a}}_b A_a = - 4\pi j_b.
\end{equation}
The benefit of doing so is that existence and uniqueness of solutions to
globally hyperbolic equations has been well studied and with a little
work we can ``extract'' solutions to our original equation.  This will be
covered in more detail below.

\subsection{Generalized Maxwell Field}

The electromagnetic field is a specific example in four dimensions
of a much broader theory of fully antisymmetric tensor fields in curved
spacetimes of arbitrary dimension.  It is mathematically natural to handle
such tensors in the language of exterior differential calculus, there
being a number of benefits to doing so: (a) All equations are coordinate
chart independent, thus we obtain global results directly. (b)  Even if a
coordinate chart is specified, the differential forms are still independent
of the choice of connection, thus substantially simplifying coordinate based
calculations.  (c) The generalized Maxwell equations do index bookkeeping
in a `natural' (albeit hidden) way for forms of different rank or in spacetimes
of varying dimension. It is precisely this mechanism by which the generalized
Maxwell equations avoid all of the consistency problems discussed above for
other types of fields. No subsidiary conditions are needed on the spacetime,
and the spacetime itself need not satisfy the Einstein equation. We will
elaborate more on this point later.

At first glance one could consider the fundamental object of study be the
$(p+1)$-form field strength $\F$, which can be thought of as a fully-%
antisymmetric rank-$(0,p+1)$ tensor. Then, the {\it generalized Maxwell
equations} are
\begin{equation}
\bd\F = 0 \qquad\mbox{ and }\qquad -\bdelta\F = \J,
\label{eq:generalized_maxwell}
\end{equation}
where $\J$ is the $p$-form current density.  Electromagnetism happens
to be the case where $\F$ is a two-form in a four-dimensional spacetime
and given a local coordinate chart, the above equations reduce to the
conventional Maxwell equations~\ref{eq:conventional_maxwell}.

However, we have chosen not to take $\F$ to be our fundamental object
for three reasons: (a) The action in terms of $\F$, as we will see below,
still involves $\A$ in the interaction term with $\J$, thus
$\A$ has to be defined from the relation $\F = \bd \A$. However, only
on spacetimes that have trivial $(p+1)$-th cohomology group, that is
$H^{p+1}(\stM) = \{ [0] \}$, can $\F$ be formulated globally in terms
of a $p$-form potential $\A\in\Omega^{p}(\stM)$ such that $\F=\bd\A$
is true everywhere. There is a topological restriction to doing this if
the cohomology group is nontrivial.  Thus starting with $\A$ as the
fundamental object avoids cohomological problems early on.
(b) Furthermore, even if we were study the free field, dropping all
the terms involving $\J$, it is unclear how to then go from the action
in terms of $\F$ to field equations without the introduction of $\A$ again.
(c) We are specifically working with the massless field here, but from
our previous experience with the Proca field we find that $\A$ is the
fundamental object of study.  Also, when $\A\in\Omega^0(\stM)$, the
action and field equation are that of the minimally-coupled massless
scalar field in curved spacetime.  So we have strong reason to treat
$\A$ as the fundamental object here which is derivable directly from
the action~(\ref{eq:the_action}) below.

This last point is rather remarkable.  The minimally coupled scalar
field and the electromagnetic field are but two examples of a general
$p$-form field theory in curved spacetime.  In some of our previous
work we had indications of this property.  Solutions of the massless
scalar field theory in four dimensions could be used to construct
the gauge photon polarizations of the (co)vector potential and that
information could be used to help generate one of the two physically
allowed polarization states of the (co)vector potential \cite{Pfen03}.
The remaining physical state comes from Gram-Schmidt orthogonalization.
It had also been noted that the quantum inequality for the electromagnetic
field was exactly twice that of the minimally-coupled scalar field.

\subsection{The generalized Maxwell field equations and fundamental solutions}
\label{sec:classical_field_eqs}

Let $\stM$ be a globally hyperbolic spacetime, that is, a manifold of
$\dim(\stM)=n$ with Lorentzian metric of signature $s=n-1$, i.e. the metric
is of the form $(+, -, -, \dots)$. On this spacetime we take our fundamental
object to be the field $\A\in\Omega^p(\stM)$ with $0\le p < n$.  The classical
action is given by
\begin{equation}
\mathcal{S}=(-1)^{p+1} \left[-\frac{1}{2}\langle \bd\A , \bd\A\rangle_\stM -
\langle\A,\J\rangle_\stM \right]+ S(\J),
\label{eq:the_action}
\end{equation}
were $S(\J)$ is the remainder of the action for the current density
$\J\in\Omega^p(\stM)$.  The only criteria that we ask of $\J$ is that
it be co-closed, i.e. $\bdelta\J = 0$ so as to preserve charge/current
conservation. Variation with respect to the field yields the
generalized Maxwell equation,
\begin{equation}
-\bdelta\bd \A = \J.
\label{eq:generalized_potential_eq}
\end{equation}

The field strength is then calculated from $\A$ by $\F =\bd\A$. It is easily
seen from this definition of the field strength that there is a {\it gauge
freedom} in that to any solution $\A$ of Eq.~(\ref{eq:generalized_potential_eq})
we may add $\bd\Lambda$  where $\Lambda\in\Omega^{p-1}(\stM)$.  While this
changes the value of the gauge field $\A$ at every point, it leaves the field
strength $\F$ unchanged.  We will denote any two solutions $\A$ and $\A'$
to be gauge equivalent by $\A \sim \A'$ if they differ by the exterior
derivative of a $(p-1)$-form. Thus when discussing the potential,
particularly for the quantum problem, we will often work in gauge equivalent
classes denoted by $[ \A ]=\A + \bd\Omega^{p-1}(\stM)$.

In practice one often chooses not to solve the above equation directly
but instead work with the constrained Klein-Gordon system,
\begin{equation}
\Box\A = \J \qquad\mbox{ with }\qquad \bdelta\A=0,
\label{eq:_Constr. KG}
\end{equation}
as any solution that satisfies~(\ref{eq:_Constr. KG}) is also a solution
to~(\ref{eq:generalized_potential_eq}).
In a given coordinate chart, this constrained Klein-Gordon system can be
written in component form as \cite{Lich61}
\begin{equation}
\nabla^\beta \nabla_\beta \A_{\alpha_1\dots\alpha_{p}} -
\sum_{j=1}^{p} {R_{\alpha_j}}^\beta \A_{\alpha_1\dots\beta\dots\alpha_{p}} +
\sum_{ j,\,k=1,\,j\neq k}^{p} {R_{\alpha_j}}^\beta
{ _{\alpha_k}}^{\gamma} \A_{\alpha_1\dots\beta\dots\gamma\dots\alpha_{p}}
= \J_{\alpha_1\dots\alpha_{p}}
\label{eq:KGincomponents}
\end{equation}
and
\begin{equation}
\nabla^\beta \A_{\beta\alpha_2\dots\alpha_{p}}=0.
\end{equation}
Here $R_{\alpha\beta}$ and
$R_{\alpha\beta\gamma\delta}$ are the Ricci and Riemann tensors, respectively.
Also, the notation is such that the index  $\beta$ in the first summation
occupies the $j^{\rm th}$ place in the tensor component of  $\A$ while in
the double summation the indices $\beta$ and $\gamma$ occupy the $j^{\rm th}$
and $k^{\rm th}$ spots in the tensor components. The Riemann and Ricci terms
are not unexpected; in differential geometry they are a result of the
Weitzenb\"{o}ck identity.

The advantage of using the constrained Klein-Gordon system of equations is
that $\Box$ is a normally hyperbolic operator \cite{BGP07} with principal
part $g^{\mu\nu}\partial_\mu\partial_\nu$.  Thus, there exists a unique
advanced (-) and retarded (+)
Green's operator denoted by  $E^{\pm}:\Omega^{p}_0(\stM)\rightarrow
\Omega^{p}(\stM)$, (see Corollary 3.4.3 of B\"{a}r, Ginoux and Pf\"{a}ffle
\cite{BGP07}, or Choquet-Bruhat \cite{Choq67} and Proposition 3.3 of
Sahlmann and Verch \cite{Sah&Ver}) which have the properties
\begin{equation}
\Box E^\pm = E^\pm \Box = \id,
\end{equation}
and for all $f\in\Omega^{p}_0(\stM)$, the $\supp(E^\pm f)\subset J^\pm(\supp(f))$.
Furthermore, the map defined by the Green's operators are sequentially continuous.
All of the differential operations on forms commute with the Green's operator;

\begin{Prop}\label{prop_commutators}
Let $f\in\Omega^p_0(\stM)$ be a test function, then the following operations
involving $E^\pm$ commute: (a) $\bd E^\pm f = E^\pm \bd f$, and
(b) $\bdelta E^\pm f = E^\pm \bdelta f$.
\end{Prop}
\begin{proof}
(a)  Let $f\in\Omega_0^p(\stM)$.  We know that $\A^\pm = E^\pm f \in \Omega^p(\stM)$
is the unique solution to $\Box \A^\pm = f$.  Likewise we have $\A'^\pm =  E^\pm
\bd f \in \Omega^{p+1}(\stM)$ is the unique solution to $\Box \A'^\pm = \bd f $.
Consider
\begin{equation}
\Box\bd\A^\pm = \bd\Box\A^\pm = \bd f = \Box\A'^\pm.
\end{equation}
Since $\A'^\pm$ is unique we deduce $\bd\A^\pm =\A'^\pm$, therefore
$\bd E^\pm f = E^\pm \bd f$.

(b)  Let $f$ and $\A^\pm$ be as defined above.  Set $\A'^\pm = E^\pm\bdelta f\in
\Omega^{p-1}(\stM)$, which is the unique solution to $\Box \A'^\pm = \bdelta f $.
Then consider
\begin{equation}
\bdelta\Box\A^\pm = \Box\bdelta\A^\pm = \bdelta f = \Box\A'^\pm.
\end{equation}
Since $\A'^\pm$ is unique we deduce $\bdelta\A^\pm =\A'^\pm$, therefore
$\bdelta E^\pm f = E^\pm \bdelta f$.
\end{proof}

We also need the advanced minus retarded propagator $E\equiv E^- - E^+$.  Since
it is a linear combination of the advanced and retarded propagators, it has all
of the same commutation properties above and gives the unique solutions to the
homogeneous (source free) Klein-Gordon equation. We are now ready to show that
the existence of a fundamental solution for the Klein-Gordon equation also gives
us a solution to the generalized Maxwell equations.

\subsection{Cauchy Problem for $\Box\A=0$}

The the initial value problem for a gauge field in four dimensional
spacetimes has been treated in the initial sections of Dimock and
for the Proca field by Furlani.  We follow closely the notation and
structure of both these papers in this section as we generalize their
results to all $p$-form fields with $p<n$ in globally hyperbolic spacetimes
of arbitrary dimension.

In order to relate initial data to solutions of the wave equation  we
first need to discuss Green's theorem for forms.
Let $\stM$ be a globally hyperbolic spacetime, and $\mathcal{O}\subset
\stM$ be an open region in the spacetime with boundary $\partial\mathcal{O}$.
Define the natural inclusion $i:\partial\mathcal{O}\rightarrow \mathcal{O}$
and $i^*$ the pullback.  On $\stM$ let $\A \in \Omega^p(\stM)$ and $\B \in
\Omega^p_0(\stM)$, then
by Stokes theorem we have
\begin{equation}
\int_\mathcal{O} \left(\A\wedge*\Box\B - \B\wedge*\Box\A\right) =
\int_{\partial\mathcal{O}} i^*\left(\A\wedge*\bd\B + \bdelta\A
\wedge*\B\right) - \int_{\partial\mathcal{O}} i^*\left(\B\wedge*\bd\A
+\bdelta\B\wedge*\A\right),
\label{eq:Greens}
\end{equation}
which is called Green's identity for $\Box$ (Sect.~7.5 of AMR \cite{AMR88}).
The integrals are all well defined because $\B$ has compact support which
does not expand under any of the derivative operations.

Next, let $\Sigma\subset\stM$ be a Cauchy surface in the spacetime and
define $\Sigma^\pm\equiv J^\pm(\Sigma)\backslash \Sigma$.  If we use $\mathcal{O}=
\Sigma^\pm$ and $\partial\mathcal{O}=\Sigma$ in Green's identity we
have
\begin{equation}
\int_{\Sigma^\pm} \left(\A\wedge*\Box\B - \B\wedge*\Box\A\right)
= \mp\left[\int_{\Sigma} i^*\left(\A\wedge*\bd\B + \bdelta\A
\wedge*\B\right) - \int_{\Sigma} i^*\left(\B\wedge*\bd\A
+\bdelta\B\wedge*\A\right)\right],
\end{equation}
where the sign difference on the right hand side comes about because
of the opposite orientation of the unit normal to the Cauchy surface.
For smooth maps, the pullback is natural with respect to both the wedge
product and the exterior derivative, thus we may distribute it across the terms above.
We define the following operations which act on $p$-forms:
\begin{eqnarray}
\rho_{(0)} &=& i^*\\
\rho_{(\bd)} &=& (-1)^{p(n-p-1)+(n-1)}*\, i^* *\bd = \rho_{(n)}\bd\\
\rho_{(\bdelta)} &=& i^*\bdelta = \rho_{(0)}\bdelta\\
\rho_{(n)} &=& (-1)^{(n-p)(p-1)+(n-1)}*\, i^* * .
\end{eqnarray}
The first operation is the pullback of the form onto the Cauchy Surface,
the second is the forward normal derivative, the third is the pullback
of the divergence and the last is the forward normal \cite{Dimock,Furl99}.
Note, all operations to the right of $i^*$ act on forms which are defined
on the whole manifold $\stM$.  However, operations to the left of $i^*$
are defined with respect to the Cauchy surface $\Sigma$, thus in the
forward normal derivative and forward normal, the first Hodge star in
each expression is with respect to the induced metric on the Cauchy
surface.  With these operations, the above equation reduces to
\begin{equation}
\int_{\Sigma^\pm} \left(\A\wedge*\Box\B - \B\wedge*\Box\A\right)
= \mp\left(\langle\rho_{(0)}\A,\rho_{(\bd)}\B\rangle_\Sigma
+\langle\rho_{(\bdelta)}\A,\rho_{(n)}\B\rangle_\Sigma
-\langle\rho_{(0)}\B,\rho_{(\bd)}\A\rangle_\Sigma
-\langle\rho_{(\bdelta)}\B,\rho_{(n)}\A\rangle_\Sigma\right).
\label{eq:applied_Greens}
\end{equation}

We begin the discussion of the fundamental solutions to the Klein-Gordon equation.
First we look at the mapping of $\Box$ solutions to initial data.

\begin{Prop}\label{Prop:A_with_f}
Let $\A\in\Omega^p(\stM)$ be a smooth solution of $\Box\A=\J$ with Cauchy
data
\begin{eqnarray}
A_{(0)}&\equiv& \rho_{(0)}\A \in\Omega^p(\Sigma)\nonumber\\
A_{(\bd)}&\equiv& \rho_{(\bd)}\A \in\Omega^p(\Sigma)\nonumber\\
A_{(\bdelta)}&\equiv& \rho_{(\bdelta)}\A \in\Omega^{p-1}(\Sigma)\nonumber\\
A_{(n)}&\equiv& \rho_{(n)}\A \in\Omega^{p-1}(\Sigma).\nonumber
\end{eqnarray}
Then, for any compactly supported test function $f\in\Omega^p_0(\stM)$
we have
\begin{equation}
\int_\stM \A\wedge*f =
\langle\J , E^- f\rangle_{\Sigma^+} +\langle\J , E^+ f\rangle_{\Sigma^-}
-\langle A_{(0)},\rho_{(\bd)} E f\rangle_\Sigma
-\langle A_{(\bdelta)},\rho_{(n)} E f\rangle_\Sigma
+\langle A_{(\bd)},\rho_{(0)} E f\rangle_\Sigma
+\langle A_{(n)},\rho_{(\bdelta)} E f\rangle_\Sigma.
\label{eq:INVprop1}
\end{equation}
\end{Prop}
\begin{proof}
This is a generalization of the proof for Theorem 5 of Furlani \cite{Furl99}.
Note, Furlani uses a different sign convention for the definition of the propagator $E$.
Using Eq.~\ref{eq:applied_Greens}, set $\A$ equal to the smooth solution of
$\Box\A = \J$ and $\B =  E^\mp f$ for $\Sigma^\pm$ respectively, then
\begin{eqnarray}
\int_{\Sigma^\pm} \left(\A\wedge*\Box E^\mp f - E^\mp f \wedge*\Box\A\right)
&=& \mp\left(\langle\rho_{(0)}\A,\rho_{(\bd)}E^\mp f\rangle_\Sigma
+\langle\rho_{(\bdelta)}\A,\rho_{(n)}E^\mp f\rangle_\Sigma\right.\nonumber\\
&&\left.-\langle\rho_{(0)}E^\mp f,\rho_{(\bd)}\A\rangle_\Sigma
-\langle\rho_{(\bdelta)}E^\mp f,\rho_{(n)}\A\rangle_\Sigma\right).
\end{eqnarray}
Substituting the Cauchy data, noting that $\A$ is a smooth
solution to the Klein-Gordon equation with source $\J$ and
using $\Box E^\pm = \id$, we can simplify the above to
\begin{equation}
\int_{\Sigma^\pm} \A\wedge* f = \langle\J , E^\mp f \rangle_{\Sigma^\pm}
\mp\left(\langle A_{(0)},\rho_{(\bd)}E^\mp f\rangle_\Sigma
+\langle A_{(\bdelta)},\rho_{(n)}E^\mp f\rangle_\Sigma
-\langle \rho_{(0)}E^\mp f, A_{(\bd)}\rangle_\Sigma
-\langle\rho_{(\bdelta)}E^\mp f, A_{(n)}\rangle_\Sigma\right)
\end{equation}
When the two above equations are added for $\Sigma^\pm$, the
result is Eq.~\ref{eq:INVprop1}.  Furthermore, the integrals in this
expression are well defined because $f$ has compact support, thus
$E^\pm f$ and $E f$ have compact support on all other Cauchy surfaces.
\end{proof}

We now address the issues of existence and uniqueness for homogenous solutions
to the Klein-Gordon equation.
\begin{Prop}\label{prop_uniqueness}
{\rm (Uniqueness of homogeneous $\Box$ solutions)} If $\A$ is a smooth solution to
$\Box\A=0$ with Cauchy data $A_{(0)}=0$, $A_{(\bdelta)} = 0$, $A_{(\bd)}=0$,
and $A_{(n)} = 0$ then $\A=0$.
\end{Prop}
\begin{proof}
By Proposition~\ref{Prop:A_with_f} we have $\int_\stM \A\wedge*f = 0$ which
is true for all compactly supported $f$; therefore, $\A=0$.
\end{proof}

\begin{Prop}\label{prop:box_existence}{\rm (Existence of homogenous $\Box$ solutions)}
Let $A_{(0)}, A_{(\bd)} \in \Omega^p_0(\Sigma)$ and $A_{(n)}, A_{(\bdelta)}
\in\Omega^{p-1}_0(\Sigma)$ specify Cauchy data on $\Sigma$.  Then
\begin{equation}
\label{eq:Distributional_solution}
\A' = E\rho_{(\bd)}' A_{(0)}+E\rho_{(n)}' A_{(\bdelta)}
-E\rho_{(\bdelta)}' A_{(n)}-E\rho_{(0)}' A_{(\bd)}
\end{equation}
is the unique smooth solution of $\Box\A'=0$ with these data.
\end{Prop}
\begin{proof}
Let $f\in\Omega_0^p(\stM)$ be any compactly supported test form and consider
\begin{eqnarray}
\langle \A', f \rangle_\stM &=&\langle E\rho_{(\bd)}' A_{(0)},f\rangle_\stM
+\langle E\rho_{(n)}' A_{(\bdelta)},f\rangle_\stM-\langle E\rho_{(\bdelta)}'
A_{(n)},f\rangle_\stM-\langle E\rho_{(0)}' A_{(\bd)},f\rangle_\stM\nonumber\\
&=& \langle \rho_{(\bd)}' A_{(0)}, E'f\rangle_\stM
+\langle \rho_{(n)}' A_{(\bdelta)},E'f\rangle_\stM-\langle \rho_{(\bdelta)}'
A_{(n)},E' f\rangle_\stM-\langle \rho_{(0)}' A_{(\bd)}, E' f\rangle_\stM.
\end{eqnarray}
The operator $\Box = -(\bdelta\bd - \bd\bdelta)$ is self-adjoint, therefore
the transpose operator $E' = -E$ (see Choquet-Bruhat \cite{Choq67}, corollary
to Theorem II).
The pullback operators are all linear and continuous, thus there exist
transpose operators for the $\rho$'s denoted with the primes and we have
\begin{equation}
\langle \A', f \rangle_\stM = -\langle A_{(0)},  \rho_{(\bd)}E f\rangle_\Sigma
-\langle A_{(\bdelta)},\rho_{(n)}Ef\rangle_\Sigma+\langle A_{(n)},\rho_{(\bdelta)}
E f\rangle_\Sigma+\langle A_{(\bd)}, \rho_{(0)}E f\rangle_\Sigma.
\end{equation}
By Proposition~\ref{Prop:A_with_f}, if $\A$ is a solution to $\Box\A=0$
with the Cauchy data given above, then we have $\langle \A', f \rangle_\stM
= \langle\A,f\rangle_\stM$ which implies $\A'=\A$ in a distributional sense.
Thus $\A'$ is identified with the unique smooth solution $\A$.
\end{proof}

From this point forward we will be discussing systems whose initial data is
smooth and compactly supported on the Cauchy
surfaces, i.e., $\left(A_{(0)}, A_{(\bd)}, A_{(n)}, A_{(\bdelta)}\right)\in
\Omega^p_0(\Sigma)\oplus\Omega^p_0(\Sigma)\oplus\Omega^{p-1}_0(\Sigma)\oplus
\Omega^{p-1}_0(\Sigma)$.  We may now discuss the sense in which $\A'$, as
defined above, varies with respect to the Cauchy data.

\begin{Prop}$\A'$ is continuously dependent on the Cauchy data $\left(A_{(0)},
A_{(\bd)}, A_{(n)}, A_{(\bdelta)}\right)$.
\end{Prop}
\begin{proof}
The proof is a generalization of Theorem 3.2.12 of B\"{a}r et.\ al.\ \cite{BGP07}.
Define $\mathcal{H}^p(\stM) \equiv \left\{ \A\in\Omega^p(\stM) | \Box\A=0 \right\}$
as the space of smooth homogeneous Klein-Gordon solutions,
then the mapping $\mathcal{P}:\mathcal{H}^p(\stM)\rightarrow\Omega^p(\Sigma)
\oplus\Omega^p(\Sigma)\oplus\Omega^{p-1}(\Sigma)\oplus\Omega^{p-1}(\Sigma)$ of
a solution to its Cauchy data is by definition both linear and continuous. Next,
let $K\subset\stM$ be a compact subset of $\stM$.  On $K$ we have the spaces
$\Omega_0^p(K) \subset \Omega^p(\stM)$ and $\Omega_0^p(K\cap\Sigma)\subset
\Omega^p(\Sigma)$ for all $p\le n$.  We also define the space $\mathcal{V}_K^p
= \mathcal{P}^{-1}\left[ \Omega^p_0(K\cap\Sigma)\oplus\Omega^p_0(K\cap\Sigma)
\oplus\Omega^{p-1}_0(K\cap\Sigma)\oplus\Omega^{p-1}_0(K\cap\Sigma)\right]$.
Since $\mathcal{P}$ is continuous and $\Omega^p_0(K\cap\Sigma)\oplus
\Omega^p_0(K\cap\Sigma)\oplus\Omega^{p-1}_0(K\cap\Sigma)\oplus\Omega^{p-1}_0
(K\cap\Sigma) \subset\Omega^p(\Sigma)\oplus \Omega^p(\Sigma)\oplus\Omega^{p-1}
(\Sigma)\oplus\Omega^{p-1}(\Sigma)$ is a closed subset, this implies
$\mathcal{V}_K^p\subset \mathcal{H}^p(\stM)$ is also a closed subset.
Furthermore, both $\Omega^p_0(K\cap\Sigma)\oplus \Omega^p_0(K\cap\Sigma)\oplus
\Omega^{p-1}_0(K\cap\Sigma)\oplus\Omega^{p-1}_0 (K\cap\Sigma)$ and
$\mathcal{V}_K^p$ are Fr\'{e}chet spaces. Also, the map $\mathcal{P}:
\mathcal{V}_K^p \rightarrow \Omega^p_0(K\cap\Sigma)\oplus\Omega^p_0(K\cap\Sigma)
\oplus \Omega^{p-1}_0(K\cap\Sigma)\oplus\Omega^{p-1}_0(K\cap\Sigma)$ is linear,
continuous, bijective, and by the open mapping theorem for Fr\'{e}chet
spaces (Theorem V.6 of Reed and Simon \cite{RSvI:80}) it is also an open
mapping.  Since a bijection that is open implies a continuous inverse, we
conclude that $\mathcal{P}^{-1}$ is continuous.

Finally, if we have a convergent sequence of Cauchy data $\left( A_{(0),i},
A_{(\bd),i}, A_{(n),i}, A_{(\bdelta),i}\right)\rightarrow \left(A_{(0)},
A_{(\bd)}, A_{(n)}, A_{(\bdelta)}\right)$ in $\Omega^p_0(\Sigma)\oplus
\Omega^p_0(\Sigma)\oplus\Omega^{p-1}_0(\Sigma)\oplus\Omega^{p-1}_0(\Sigma)$,
then we can choose a compact subset $K \subset \stM$ with the property
that $\mbox{supp}\left( A_{(0),i}\right)\cup\mbox{supp}\left( A_{(\bd),i}\right)
\cup \mbox{supp}\left( A_{(n),i}\right)\cup\mbox{supp}\left( A_{(\bdelta),i}
\right) \subset K$ for all $i$ and $\mbox{supp}\left( A_{(0)}\right)\cup
\mbox{supp}\left( A_{(\bd)}\right)\cup \mbox{supp}\left( A_{(n)}\right)\cup
\mbox{supp}\left( A_{(\bdelta)}\right)\subset K$.  Thus, $\left( A_{(0),i},
A_{(\bd),i}, A_{(n),i}, A_{(\bdelta),i}\right)\rightarrow \left(A_{(0)},
A_{(\bd)}, A_{(n)}, A_{(\bdelta)}\right)$ in $\Omega^p_0(K\cap\Sigma)\oplus
\Omega^p_0(K\cap\Sigma)\oplus\Omega^{p-1}_0(K\cap\Sigma)\oplus\Omega^{p-1}_0
(K\cap\Sigma)$ and we conclude that the inverse mapping $\mathcal{P}^{-1}
\left( A_{(0),i},A_{(\bd),i}, A_{(n),i}, A_{(\bdelta),i}\right)\rightarrow
\mathcal{P}^{-1}\left(A_{(0)},A_{(\bd)}, A_{(n)}, A_{(\bdelta)}\right)$.
\end{proof}

We can summarize the three above results in a single theorem regarding the
Cauchy problem for the Klein-Gordon Equation;

\begin{Thm}\label{prop:_Cauchy for Box}{\rm (Well-posed Cauchy problem for $\Box\A=0$.)}
Let $\stM$ be a globally hyperbolic Lorentzian spacetime with Cauchy
surface $\Sigma\subset\stM$. Let $(A_{(0)}, A_{(\bd)}, A_{(n)},
A_{(\bdelta)})\in\Omega^p_0(\Sigma)\oplus\Omega^p_0(\Sigma)\oplus
\Omega^{p-1}_0(\Sigma)\oplus \Omega^{p-1}_0(\Sigma)$ specify initial
data on $\Sigma$.  Then, $\A'$ given by Equation~\ref{eq:Distributional_solution} is
the unique smooth solution to $\Box\A'=0$ which satisfies
\begin{equation}
\rho_{(0)}\A' = A_{(0)},\qquad\rho_{(\bd)}\A' = A_{(\bd)},\qquad
\rho_{(n)}\A' = A_{(n)},\qquad\mbox{and}\qquad\rho_{(\bdelta)}\A' = A_{(\bdelta)}.
\end{equation}
Furthermore, $\A'$ is continuously dependent on the Cauchy data.
\end{Thm}
This theorem is the generalization to all $p$-forms of Furlani's Theorem~1
for 1-forms \cite{Furl99}.

\subsection{Subspaces of Solutions to $\Box\A = 0$}

It turns out that the  solution space to the homogenous Klein-Gordon
equation is extremely rich.  In fact, it is too large for
the purposes of solving the Maxwell equation.  Amongst all the possible solutions
to the homogenous Klein-Gordon equation, it is possible to identify the subspace
of solutions to the constrained Klein-Gordon system~\ref{eq:_Constr. KG} by its
Cauchy data;

\begin{Prop}\label{thm_Maxwell_solns} {\rm (Maxwell Solutions)}
Suppose $\A\in\Omega^p(\stM)$ solves $\Box\A = 0$ with Cauchy data
$(A_{(0)}, A_{(\bd)}, A_{(n)}, A_{(\bdelta)})$; then $-\bdelta\bd\A = 0$
if and only if $\bdelta A_{(\bd)} = 0$ and $\bd A_{(\bdelta)}=0$.
\end{Prop}
\begin{proof}
First, suppose that $\A$ solves both $\Box\A = 0$ and $-\bdelta\bd\A=0$,
and evaluate the two conditions,
\begin{eqnarray}
\bdelta A_{(\bd)} &=& \bdelta \rho_{(\bd)} \A = \bdelta \rho_{(n)}\bd\A
= (-1)^{n(p+1)}\rho_{(n)}\bdelta\bd\A = 0,\\
\bd A_{(\bdelta)} &=& \bd \rho_{(\bdelta)} \A = \bd \rho_{(0)} \bdelta\A
= \rho_{(0)}\bd\bdelta\A = - \rho_{(0)}\bdelta\bd\A = 0.
\end{eqnarray}
Alternatively, suppose we have Cauchy data $(A_{(0)}, A_{(\bd)}, A_{(n)},
A_{(\bdelta)})$ with $\bdelta A_{(\bd)} = 0$ and $\bd A_{(\bdelta)}=0$. By
Theorem~\ref{prop:_Cauchy for Box}, there exists a unique
$\A$ which is a solutions to $\Box\A = 0$.  Next, let $f\in\Omega^p_0(\stM)$,
then by Proposition~\ref{Prop:A_with_f} we can evaluate
\begin{eqnarray}
\langle -\bdelta\bd\A,f\rangle_\stM &=& -\langle \A, \bdelta\bd f\rangle_\stM
\nonumber\\
&=&\langle A_{(0)},\rho_{(\bd)} E\bdelta\bd f\rangle_\Sigma
+\langle A_{(\bdelta)},\rho_{(n)} E\bdelta\bd f\rangle_\Sigma
-\langle A_{(n)},\rho_{(\bdelta)} E\bdelta\bd f\rangle_\Sigma
-\langle A_{(\bd)},\rho_{(0)} E\bdelta\bd f\rangle_\Sigma\nonumber\\
&=&-\langle A_{(0)},\rho_{(\bd)} \bd\bdelta E f\rangle_\Sigma
+\langle A_{(\bdelta)},\rho_{(n)} \bdelta\bd E f\rangle_\Sigma
-\langle A_{(n)},\rho_{(\bdelta)} \bdelta\bd E f\rangle_\Sigma
+\langle A_{(\bd)},\rho_{(0)} \bd\bdelta E f\rangle_\Sigma,
\end{eqnarray}
where we have made use of $\Box\A = 0$ in the first and
fourth terms.  The first and third terms vanish because $\rho_{(\bd)}\bd = 0$
and $\rho_{(\bdelta)}\bdelta = 0$.  For the remaining two terms, we
commute the pullback operators with the first derivative operation to
obtain
\begin{equation}
\langle -\bdelta\bd\A,f\rangle_\stM = (-1)^{n(p+1)}\langle \bd A_{(\bdelta)}
, \rho_{(\bd)} E f \rangle_\Sigma + \langle \bdelta A_{(\bd)}, \rho_{(\bdelta)}
E f \rangle_\Sigma = 0,
\end{equation}
where we have made use of the properties of the Cauchy data.  This is true
for all $f$, so we conclude $-\bdelta\bd\A=0$.
\end{proof}

The above choice of Cauchy data is sufficient to
ensure $\bd\bdelta\A=0$, but this does not imply that $\bdelta\A=0$. The
subspace of Lorenz solutions, contained within the space of Maxwell solutions,
requires further restrictions on the Cauchy data;

\begin{Prop}{\rm (Lorenz solutions)}\label{prop:Lorenz_soln}
Suppose $\A\in\Omega^p(\stM)$ solves $\Box\A=0$ with Cauchy data $\left ( A_{(0)},
A_{(\bd)}, A_{(n)}, A_{(\bdelta)} \right)$, then $\bdelta\A = 0$ if and only if
$\bdelta A_{(\bd)} = 0$, $\bdelta A_{(n)} = 0$, and $A_{(\bdelta)}=0$.
\end{Prop}
\begin{proof}
First suppose that $\A$ solves both $\Box\A = 0$ and $\bdelta\A =0$ and let us
evaluate the three conditions above.  For the forward normal derivative we have
\begin{equation}
\bdelta A_{(\bd)} = \bdelta \rho_{(\bd)}\A = \bdelta \rho_{(n)}\bd \A.
\end{equation}
Using the identity for $p$-forms that $\bdelta\rho_{(n)} = (-1)^{np} \rho_{(n)}\bdelta $
we find
\begin{equation}
\bdelta A_{(\bd)} = (-1)^{n(p+1)}\rho_{(n)}\bdelta\bd\A = -(-1)^{n(p+1)}\rho_{(n)}\bd\bdelta\A=0,
\end{equation}
where we have used the wave equation and the subsidiary condition in the last two
steps. For the divergence of the forward normal we find
\begin{equation}
\bdelta A_{(n)} = \bdelta\rho_{(n)} \A = (-1)^{np}\rho_{(n)}\bdelta\A =0.
\end{equation}
Finally, for the pullback of the divergence we have
\begin{equation}
A_{(\bdelta)} = \rho_{(0)}\bdelta\A = 0.
\end{equation}

On the other hand, assume $\A$ is a Klein-Gordon solution whose Cauchy
data satisfies $\bdelta A_{(\bd)} = 0=\bdelta A_{(n)}$ and $A_{(\bdelta)}=0$.
Let $f\in\Omega_0^{p-1} (\stM)$ and evaluate
\begin{eqnarray}
\langle \bdelta\A , f \rangle_\stM &=& \langle \A, \bd  f \rangle_\stM\nonumber\\
&=& -\langle A_{(0)}, \rho_{(\bd)} E \bd f \rangle_\Sigma
- \langle A_{(\bdelta)}, \rho_{(n)} E \bd f \rangle_\Sigma
+\langle A_{(\bd)}, \rho_{(0)} E \bd f \rangle_\Sigma
+ \langle A_{(n)}, \rho_{(\bdelta)} E \bd f \rangle_\Sigma\nonumber\\
&=& -\langle A_{(0)}, \rho_{(\bd)} \bd  E f \rangle_\Sigma
- \langle 0, \rho_{(n)} \bd  E f \rangle_\Sigma
+\langle A_{(\bd)}, \rho_{(0)} \bd E f \rangle_\Sigma
+ \langle A_{(n)}, \rho_{(\bdelta)} \bd E f \rangle_\Sigma,
\end{eqnarray}
where we have used the fact that $\bd$ commutes with $E$ on forms of compact
support. The first term vanishes because $\rho_{(\bd)}\bd = 0$. The second
term trivially vanishes. For the third term we have that $\rho_{(0)}$ and
$\bd$ commute. For the fourth term we use $\rho_{(\bdelta)} = \rho_{(0)}\bdelta$
and the fact that $E f$ is a solution to the homogeneous Klein-Gordon equation
to swap the order of the derivative operators. Therefore
\begin{equation}
\langle \bdelta\A , f \rangle_\stM =
\langle A_{(\bd)}, \bd \rho_{(0)} E f \rangle_\Sigma
-\langle  A_{(n)}, \rho_{(0)} \bd\bdelta E f \rangle_\Sigma\
=
\langle  \bdelta A_{(\bd)}, \rho_{(0)}  E f \rangle_\Sigma
- \langle \bdelta A_{(n)}, \rho_{(0)} \bdelta E f \rangle_\Sigma
=0,
\end{equation}
where we have used the conditions on the initial data.
Since this is true for all $f$ we deduce that $\bdelta\A = 0$.
\end{proof}

Contained within the Lorenz solutions are the Coulomb solutions which have
Lorenz solution Cauchy data with the additional constraint $A_{(n)}=0$,
or more succinctly,  $(A_{(0)}, A_{(\bd)}, 0 ,0)$ with $\bdelta A_{(\bd)}
= 0$.  Coulomb solutions are particulary useful as
representatives of the equivalence class of physically realizable field
configurations.

By Proposition~\ref{thm_Maxwell_solns} we can infer that the constrained
Klein-Gordon system is self consistent in curved spacetimes of arbitrary
dimension. A slightly different way to see this is to begin with the
evolution equation $\Box\A = \J$ and take $\bdelta$ of it, yielding
\begin{equation}
0 = \bdelta\J = \bdelta\Box\A = \Box\bdelta\A.
\end{equation}
We observe that $\bdelta\A$ satisfies the source free Klein-Gordon equation.
If the Cauchy data for $\bdelta\A$ vanishes on the initial Cauchy
surface then by Proposition~\ref{prop_uniqueness} the unique solution
is $\bdelta\A=0$.  Similar arguments also hold for Maxwell and Coulomb
gauge solutions. For Lorenz solutions, what ensures this property are
the Riemann and Ricci terms in Eq.~(\ref{eq:KGincomponents}).  They commute
with the divergence in the proper way as a result of the first and second
Bianchi identities.  Unlike the results of Buchdahl and Higuchi for spinor
and massive symmetric tensor fields, no other conditions are required
of the spacetime for the $p$-form fields satisfying the generalized
Maxwell equation.  In fact, the spacetime does not need to satisfy
the Einstein equations in any way.  At a deeper level
what we have really done is to take the flat space field equations and first
make the minimal substitution.  Only afterward do we then commute
the covariant derivatives.  This gives rise to all of the Riemann and
Ricci terms.  The beauty of using exterior calculus is all this is handled
without our having to do it explicitly.

In addition to the Lorenz solutions, the Maxwell solutions also contains
a subspace of the closed solutions which can be identified by their
Cauchy data;

\begin{Prop} \label{thm: gauge_soln}{\rm (Closed Solutions)}
Suppose $\A\in\Omega^p(\stM)$ solves $\Box\A = 0$ with Cauchy data
$(A_{(0)}, A_{(\bd)}, A_{(n)}, A_{(\bdelta)})$; then $\A\in{\rm ker}\,
\bd_p$ on $\stM$ if and only if $\bd A_{(0)} = 0$, $A_{(\bd)}=0$
and $\bd A_{(\bdelta)}=0$.
\end{Prop}

\begin{proof}
Let $\A\in{\rm ker}\,\bd_p$ be a Klein-Gordon solution, then $\A$ is
also a Maxwell solution.  Next, evaluate the three conditions;
\begin{eqnarray}
\bd A_{(0)} &=& \bd\rho_{(0)}\A = \rho_{(0)}\bd\A = 0,\nonumber\\
 A_{(\bd)} &=& \rho_{(n)}\bd\A = 0,\nonumber\\
\bd A_{(\bdelta)} &=& \bd\rho_{(0)}\bdelta\A = \rho_{(0)}\bd\bdelta\A
=-\rho_{(0)}\bdelta\bd\A = 0.
\end{eqnarray}

Alternatively, assume Cauchy data $(A_{(0)}, 0, A_{(n)}, A_{(\bdelta)})$
where $\bd A_{(0)} = 0 =\bd A_{(\bdelta)}$.  Then, there exists an unique
$\A$ that solves $\Box\A = 0$ with this Cauchy data.  By proposition~%
\ref{thm_Maxwell_solns} above, it is also a Maxwell solution. Additionally,
let $g\in\Omega^{p+1}_0(\stM)$ and evaluate
\begin{eqnarray}
\langle \bd \A , g\rangle_\stM &=& \langle \A, \bdelta g \rangle_\stM \nonumber\\
&=& -\langle A_{(0)},\rho_{(\bd)}E\bdelta g \rangle_\Sigma -\langle A_{(\bdelta)}
,\rho_{(n)} E \bdelta g \rangle_\Sigma + \langle A_{(n)}, \rho_{(\bdelta)}
E\bdelta g\rangle_\Sigma + \langle 0, \rho_{(0)} E\bdelta g\rangle_\Sigma
\nonumber\\
&=& -\langle A_{(0)},\rho_{(n)}\bd\bdelta E g \rangle_\Sigma -
\langle A_{(\bdelta)},\rho_{(n)} \bdelta E g \rangle_\Sigma + \langle
A_{(n)}, \rho_{(\bdelta)}\bdelta E g\rangle_\Sigma\nonumber\\
&=&\langle A_{(0)}, \rho_{(n)}\bdelta\bd E g\rangle_\Sigma
-\langle A_{(\bdelta)},\rho_{(n)}\bdelta E g\rangle_\Sigma\nonumber\\
&=& (-1)^{n(p+2)}\langle \bd A_{(0)} , \rho_{(n)} \bd E g\rangle_\Sigma-
(-1)^{n(p+1)}\langle \bd A_{(\bdelta)}, \rho_{(n)} E g\rangle_\Sigma\nonumber\\
&=& 0,
\end{eqnarray}
which is true for all $g$, implying that $\bd \A =0$.  Therefore,
$\A\in{\rm ker}\, \bd_p$ on the spacetime $\stM$.
\end{proof}

\noindent{\it Remarks}:  There does not appear to be any trivial way
to further distinguish the space of exact solutions within the space
of closed solutions.  However, it is possible to identify the ``traditional''
harmonic solutions which satisfy $\Box\A=0$ with $\bd\A=0$ and $\bdelta\A=0$.
\begin{Prop}{\rm (``Traditional'' Harmonic Solutions)}
Suppose $\A\in\Omega^p(\stM)$ solves $\Box\A=0$ with Cauchy data
$(A_{(0)}, A_{(\bd)}, A_{(n)}, A_{(\bdelta)})$; then $\bd\A=\bdelta\A=0$
if and only if $\bd A_{(0)}=0$, $A_{(\bd)}=0$, $\bdelta A_{(n)}=0$ and $A_{(\bdelta)}=0$.
\end{Prop}
\begin{proof}
First, supposed $\A$ satifies $\bd\A=0$ and $\bdelta\A=0$, then $\A$ trivially satisfies the
Klein-Gordon and Maxwell equations.  We evaluate
\begin{eqnarray}
\bd A_{(0)} &=& \bd\rho_{(0)}\A = \rho_{(0)}\bd\A = 0,\nonumber\\
A_{(\bd)} &=& \rho_{(n)}\bd\A = 0,\nonumber\\
\bdelta A_{(n)} &=& \bdelta\rho_{(n)}\A = (-1)^{np}\rho_{(n)}\bdelta\A = 0,
\nonumber\\
\A_{(\bdelta)} &=& \rho_{(0)}\bdelta\A = 0.
\end{eqnarray}

Alternatively, assume Cauchy data $(A_{(0)}, 0, A_{(n)}, 0)$ where $\bd A_{(0)} = 0$
and $\bdelta A_{(n)} =0$.  Then, there exists an unique $\A$ that solves $\Box\A = 0$
with this Cauchy data.  By Proposition~\ref{prop:Lorenz_soln} above, it is also a
Lorenz solution with $\bdelta\A=0$.  By Proposition~\ref{thm: gauge_soln},
$\A\in{\rm ker}\, \bd_p$, i.e., $\bd\A=0$.
\end{proof}

\begin{figure}[]
      \begin{center}
      \psset{xunit=14mm,yunit=12mm,runit=14mm}
      \begin{pspicture}(-3,-2.5)(3,3)

      \psellipse[linewidth=2pt](0,.5)(3.5,3)
      \psellipse[linewidth=2pt,fillstyle=vlines,hatchcolor=green](1,0)(2,2)
      \psellipse[linewidth=2pt,fillstyle=hlines,hatchcolor=orange](-1,0)(2,2)
      \psellipse[linewidth=2pt](1,0)(2,2)

      \rput*(2,0){{\begin{tabular}{c}\bf Lorenz\\  $(\bdelta\A=0)$ \end{tabular}}}
      \rput*(-2,0){{\begin{tabular}{c}\bf Closed\\  $(\bd\A=0)$ \end{tabular}}}
      \rput*(0,0){{\begin{tabular}{c}\bf Traditional\\ \bf Harmonic\\ \bf Solutions\end{tabular}}}
      \rput(0,2.6){{\begin{tabular}{c}\bf Maxwell Solutions\\ $(\Box\A=0\mbox{ with }\bd\bdelta\A=0)$ \end{tabular}}}

      \end{pspicture}
      \end{center}
      \begin{caption}
      {A graphical representation of the various solutions spaces to the Maxwell equations.  Not shown are the
      Coulomb solutions which would be contained within the Lorenz solutions.}
      \label{fig:solution_diagram}
      \end{caption}
\end{figure}

We caution the reader, the well known result in Hodge-DeRham theory \cite{AMR88} that
$\Delta\alpha = 0$ iff $\bdelta\alpha = 0$ and $\bd\alpha=0$, for $\alpha$ a p-form on a
compact Riemannian manifold is not generally true for a p-form on a Lorentzian spacetime, i.e.,
$\Box\A = 0$ does not immediately imply $\bdelta\A = 0$ and $\bd\A=0$, as evidenced in the cases
above.
The multiple subspaces that are solutions to the Maxwell equation are graphically represented
in Fig.~\ref{fig:solution_diagram}.

\subsection{Cauchy problem for $-\bdelta\bd\A = 0$}

Because of the gauge freedom in the Maxwell equation, the space of Maxwell solutions is still
too large.  It turns out that to specify a unique field-strength $\F = \bd A$, we only need a
fraction of the entire Maxwell solutions.  It is found that we only need to use a reduced set of
Cauchy-data, $(A_{(0)}, A_{(\bd)})$.  For this reduced Cauchy data we are prepared to discuss the
existence of solutions to the Maxwell equation;

\begin{Prop}{\rm (Existence of Homogenous Maxwell Solutions)}
For any $( A_{(0)}, A_{(\bd)} )\in\Omega^{p}_0(\Sigma)\oplus\Omega^{p}_0(\Sigma)$
with $\bdelta A_{(\bd)} = 0$, there exists an $\A\in\Omega^p(\stM)$ such that
\begin{equation}
-\bdelta\bd\A = 0,\qquad\qquad\rho_{(0)}\A = A_{(0)}\qquad\mbox{and}\qquad
\rho_{(\bd)}\A = A_{(\bd)}.
\end{equation}
\end{Prop}
\begin{proof}
This is a broader proof than that given by Proposition 2.10 of Pfenning (original version
of this manuscript), which was itself a generalization of the proof given for Dimock's
Proposition~2. We add additional constraints and show that a solution still exists.
Choose any $( A_{(n)}, A_{(\bdelta)} )\in\Omega^{p-1}_0(\Sigma)\oplus\Omega^{p-1}_0
(\Sigma)$ with $\bd A_{(\bdelta)}=0.$  We now have Cauchy data
\begin{equation}
(A_{(0)}, A_{(\bd)},A_{(n)}, A_{(\bdelta)})\qquad\mbox{with}\qquad
\bdelta A_{(\bd)} = 0\qquad\mbox{and}\qquad \bd A_{(\bdelta)}=0.
\end{equation}
By Theorem~\ref{prop:_Cauchy for Box}, there exists an unique $\A\in\Omega^p
(\stM)$ which is a $\Box\A = 0$ solution with this Cauchy data.  Furthermore, by
proposition~\ref{thm_Maxwell_solns}, $\A$ also satisfies $-\bdelta\bd\A = 0$.
\end{proof}

For scalar fields $\A\in\Omega^0(\stM)$, this is the end of the story, as the
only non-zero portions of the Cauchy data are $( A_{(0)}, A_{(\bd)} )\in
\Omega^{0}_0(\Sigma)\oplus\Omega^{0}_0(\Sigma)$.  Thus, every $0$-form
solution to the generalized Maxwell equation, and every resulting field
strength $\F=\bd\A$ is uniquely determined by the reduced Cauchy data alone.
It is straightforward to see from the text after Proposition~\ref{prop:Lorenz_soln}
that every $0$-form solution is, in fact, a Coulomb solution.

Unfortunately, the same is not true for $p$-form theories with rank $p\geq1$;
specifying only the reduced Cauchy data $( A_{(0)}, A_{(\bd)} )$
destroys uniqueness, which requires the specification of all four elements.
However, $( A_{(0)}, A_{(\bd)} )$ is sufficient to specify an equivalence
class of solutions, for which every element of the equivalence class yields
the same field strength $\F$.

The traditional gauge freedom in electromagnetism is that two vector
potentials are considered equivalent if they differ by the gradient of a
scalar.  For the $p$-form theory in curved spacetime, two Maxwell solutions
$\A,\A' \in\Omega^p(\stM)$ are gauge-equivalent, denoted $\A'\sim\A$, if
they differ by an exact form $\bd\lambda$ where $\lambda\in\Omega^{p-1}(\stM)$.
An equivalence class is denoted by $[\A]$ where $\A$ is a representative
of the class. It is easily shown that all elements of an equivalence class
lead to the same field strength, i.e., if $\A,\A' \in[\A]$ then
\begin{equation}
\F' = \bd\A' = \bd(\A+ \bd\lambda) = \bd\A = \F.
\end{equation}
The equivalence of Maxwell solutions under the quotienting schemes can
be related to the Cauchy data of the solutions.  Here we give the correction of
Proposition~2.13 of the original version of this manuscript;
\begin{Prop}{\rm (Gauge Equivalence of Maxwell Solutions)}\label{prop: Trad gauge}\\
\indent(a) Let $( A_{(0)}, A_{(\bd)} )\in\Omega^{p}_0(\Sigma)\oplus\Omega^{p}_0(\Sigma)$
with $\bdelta A_{(\bd)} = 0$ specify reduced Cauchy data on $\Sigma$.  If $\A,\A'
\in\Omega^p(\stM)$ are Maxwell solutions with this data, then $\A\sim\A'$.

(b) Let $(A_{(0)},A_{(\bd)})$ and $(A'_{(0)},A'_{(\bd)})$ with $\bdelta A_{(\bd)}
= 0 = \bdelta A'_{(\bd)}$ specify Cauchy data on a common Cauchy surface $\Sigma$.
If $\A,\A' \in\Omega^p(\stM)$ are Maxwell solutions with these Cauchy data, then
$\A'\sim\A$ if and only if $A'_{(0)}\sim A_{(0)}$, $A'_{(\bd)}=A_{(\bd)}$.
\end{Prop}
\begin{proof}
(a) In terms of the complete Cauchy data on $\Sigma$, any Maxwell solution will have
\begin{equation}
\A \mapsto (A_{(0)},A_{(\bd)},A_{(n)},A_{(\bdelta)})
\qquad\mbox{ and }\qquad
\A' \mapsto (A_{(0)},A_{(\bd)},A'_{(n)},A'_{(\bdelta)})
\end{equation}
with the constraints $\bdelta A_{(\bd)}=0$ and $\bd A_{(\bdelta)}=0=\bd A'_{(\bdelta)}$.

On the Cauchy surface $\Sigma$, we specify complete initial data $(0, 0, A'_{(n)}-A_{(n)},
A'_{(\bdelta)}-A_{(\bdelta)})\in\Omega^{p}_0(\Sigma)\oplus\Omega^{p}_0(\Sigma)\oplus
\Omega^{p-1}_0(\Sigma)\oplus\Omega^{p-1}_0(\Sigma)$.
By Proposition~\ref{prop:_Cauchy for Box}, there exists a unique smooth $\B\in\Omega^{p}
(\stM)$ which is a solution to $\Box\B=0$ with this Cauchy data.  By
Propositions~\ref{thm_Maxwell_solns} and \ref{thm: gauge_soln}, $\B$ is a closed solution
to the Maxwell equation.

Define $\widetilde{\A}= \A +\B$, which is also a Maxwell solution, and evaluate its Cauchy data;
\begin{eqnarray}
\rho_{(0)} \widetilde{\A} &=& \rho_{(0)}(\A+\B) = A_{(0)} + \rho_{(0)}\B
= A_{(0)},\nonumber\\
\rho_{(\bd)} \widetilde{\A} &=& \rho_{(\bd)}(\A+\B) = A_{(\bd)} + \rho_{(\bd)}
\B = A_{(\bd)},\nonumber\\
\rho_{(n)} \widetilde{\A} &=& \rho_{(n)}(\A + \B) = A_{(n)} + \rho_{(n)}\B
= A_{(n)}+A'_{(n)}-A_{(n)} = A'_{(n)},\nonumber\\
\rho_{(\bdelta)} \widetilde{\A} &=& \rho_{(\bdelta)}(\A+\B) = A_{(\bdelta)} +
\rho_{(\bdelta)}\B = A_{(\bdelta)} + A'_{(\bdelta)}-A_{(\bdelta)} =
A'_{(\bdelta)}.
\end{eqnarray}
We immediately see that $\A'$ and $\widetilde{\A}$ are both Maxwell solutions with identical
complete Cauchy data, and by uniqueness of homogenous Klein-Gordon solutions $\A' = \widetilde{\A}
= \A +\B$. Setting the $p$-th deRham cohomology group for the spacetime, $H^p(\stM)$, to be trivial,
we can then conclude that $\B$ is exact.  More specifically, there exists a $\lambda\in\Omega^{p-1}
(\stM)$ such that $\B=\bd\lambda$.  Upon substitution, we have $\A' = \A +\bd\lambda$, or
more simply, $\A'\sim\A$.

(b)Let $\lambda\in\Omega^{p-1}(\stM)$ and set $\A' = \A + \bd\lambda$;  then it
immediately follows $A'_{(0)}\sim A_{(0)}$ and $A'_{(\bd)}=A_{(\bd)}$.

Alternatively,  assume $A'_{(0)} = A_{(0)} +\bd \alpha$ and $A'_{(\bd)}=A_{(\bd)}$,
where $\alpha\in\Omega^{p-1}_0(\Sigma)$.  By part (a) of this proposition, $\A\sim
\widetilde{\A}$, where $\widetilde{\A}$ is the Coulomb solution with complete Cauchy
data $(A_{(0)},A_{(\bd)},0,0)$.  Likewise, $\A'\sim\widetilde{\A}'$ where $\widetilde{\A}'$
is the Coulomb solution with complete Cauchy data $(A_{(0)}+\bd\alpha,A_{(\bd)},0,0)$.

Next, on the Cauchy surface $\Sigma$, specify initial data $(\alpha, 0,0,0)\in
\Omega^{p-1}_0(\Sigma)\oplus\Omega^{p-1}_0(\Sigma)\oplus \Omega^{p-2}_0(\Sigma)
\oplus\Omega^{p-2}_0(\Sigma)$.  By Theorem~\ref{prop:_Cauchy for Box}, there
exists a unique smooth $\lambda$ which is a solution to $\Box\lambda=0$ with
this Cauchy data.  Furthermore, by Proposition~\ref{prop:Lorenz_soln} and the
text following the proposition, $\lambda$ is a Coulomb solution to the Maxwell
equation.

Define $\B = \widetilde{\A} + \bd\lambda$, then
\begin{equation}
\Box\B = \Box\widetilde{\A} + \Box\bd\lambda = \bd\Box\lambda = 0
\end{equation}
and
\begin{equation}
\bdelta\B = \bdelta\widetilde{\A} +\bdelta\bd\lambda = 0.
\end{equation}
Evaluating the initial data for $\B$, we find;
\begin{eqnarray}
\rho_{(0)}\B &=& \rho_{(0)}\widetilde{\A} + \bd\rho_{(0)}\lambda=A_{(0)}
+\bd\alpha, \nonumber\\
\rho_{(\bd)}\B &=& \rho_{(\bd)}\widetilde{\A} + \rho_{(\bd)}\bd\lambda =
A_{(\bd)}, \nonumber\\
\rho_{(n)}\B &=&  \rho_{(n)}\widetilde{\A} + \rho_{(n)}\bd\lambda =
\rho_{(\bd)}\lambda = 0,\nonumber\\
\rho_{(\bdelta)}\B &=&\rho_{(\bdelta)}\widetilde{\A} + \rho_{(\bdelta)}\bd\lambda =
\rho_{(0)}\bdelta\bd\lambda = 0.
\end{eqnarray}
Thus, $\B$ is a Coulomb solution with complete Cauchy data identical to that
$\widetilde{\A}'$. By the uniqueness of solutions to the Klein-Gordon equation,
we have $\widetilde{\A}' = \B = \widetilde{\A} + \bd\lambda$, or more
simply $\widetilde{\A'}\sim\widetilde{\A}$.  Therefore, we conclude $\A'\sim\A$.
\end{proof}

It is unfortunate that we need to force the cohomology constraint on the spacetime
in part (a) of the proposition above;
a requirement we hope to lift in future work.  We are now prepared to show that
any solution to the Maxwell equation is gauge equivalent to a fundamental solution.

\begin{Thm}{\rm(Gauge Equivalence to Fundamental Solutions)}\label{prop:fund_soln_gauge}\\
(a) Let $\J\in\Omega^p(\stM)$ with $\bdelta\J = 0$ and $\mbox{supp}(\J)$ compact to
the past/future, then $\A^\pm = E^\pm\J$ solves $-\bdelta\bd\A^\pm = \J$.\\
(b) If $\A^\pm\in\Omega^p(\stM)$, $\mbox{supp}(\A^\pm)$ is compact to the
past/future and $-\bdelta\bd\A^\pm = \J$ (so $\bdelta\J=0$ and $\mbox{supp}(\J)$
compact to the past/future) then $\A^\pm \sim E^\pm\J$.\\
(c) $\A\in\Omega^p(\stM)$ satisfies $-\bdelta\bd\A = 0$ on spacetimes with
compact spacelike Cauchy surfaces if and only if $\A\sim E\J$ for some
$\J\in\Omega^p_0(\stM)$ with $\bdelta\J = 0$.
\end{Thm}
\begin{proof}
Part (a) is proven directly in Proposition 4 by Dimock \cite{Dimock}, to which we
refer the reader.  Parts (b) and (c) are generalizations of the corresponding
parts of proposition 4 in Dimock and we leave the proof to the reader.
\end{proof}

\subsection{Classical Phase Space}

Finally, as a precursor to quantization we discuss the classical phase
space for the $p$-form field which consists of a vector space and a
non-degenerate antisymmetric bilinear form.  For the most part this section
closely follows the discussion of the phase space for the classical
electromagnetic field found in Dimock \cite{Dimock}. All of the propositions
in his Section 3 trivially generalize from one forms to p-forms, so we will
therefore be very brief.  In addition, to simplify further discussion, we will
also assume from this point forward that the Cauchy surface is compact, thus
$\Omega^p(\Sigma)=\Omega^p_0(\Sigma)$.  We conjecture that our analysis
can be appropriately extended to spacetimes with noncompact Cauchy surfaces.

In a typical Hamiltonian formulation, the Cauchy data for the field is
specified on some ``constant time'' hypersurface.  The field's then
evolves according to the flow generated by Hamilton's equations.  We could
use the complete set of Cauchy data for our p-form field, but as we have
seen above, the Cauchy problem is well posed on gauge equivalent classes.
Therefore the initial formulation of the phase space can be accomplished
via the vector space where points are  $\left( A_{(0)},A_{(\bd)} \right)\in
\mathcal{P}_0(\Sigma):=\Omega^p_0(\Sigma)\times \Omega^p_0(\Sigma)$, i.e., the part
of the Cauchy data (with compact support) which can not be gauge transformed away.
Then on $\mathcal{P}_0(\Sigma) \times\mathcal{P}_0(\Sigma)$ define the antisymmetric
bilinear form
\begin{equation}
\sigma_{\Sigma} \left( A_{(0)}, A_{(\bd)} ; B_{(0)}, B_{(\bd)} \right) =
\langle  A_{(0)}, B_{(\bd)}\rangle_\Sigma - \langle B_{(0)}, A_{(\bd)}\rangle_\Sigma.
\end{equation}
Unfortunately this form is degenerate because for all $B_{(\bd)}$ with $\bdelta B_{(\bd)}
=0$ we have $\langle\bd\chi, B_{(\bd)}\rangle_\Sigma = \langle\chi, \bdelta B_{(\bd)}
\rangle_\Sigma = 0$ even though $\bd\chi\neq 0$.

The way to remove the degeneracy is to pass to gauge equivalent classes of
Cauchy data with points being given by the pair $( [A_{(0)}], A_{(\bd)} )\in
\mathcal{P}:=\Omega^p_0(\Sigma)/d\Omega^{p-1}_0(\Sigma)\times \Omega^p_0(\Sigma)$
then
\begin{equation}
\sigma_{\Sigma} \left([ A_{(0)}], A_{(\bd)} ; [B_{(0)}], B_{(\bd)} \right) =
\langle  [A_{(0)}], B_{(\bd)}\rangle_\Sigma - \langle [B_{(0)}], A_{(\bd)}
\rangle_\Sigma.
\end{equation}
is a suitable weakly non-degenerate bilinear form, as proven in Dimock's
Proposition 5 \cite{Dimock} .

Given any set of Cauchy data in $\mathcal{P}$, we know from the preceding
section that there is a unique equivalence class of solutions to the
Maxwell equations with this Cauchy data. Therefore, we can reformulated our
phase space  to include time evolution without the specific introduction of a
Hamiltonian.  Define the solution space of all real valued gauge equivalent
Maxwell solutions with Cauchy data on $\Sigma$ as
\begin{equation}
\mathcal{M}^p (\stM)\equiv {\left\{ \A\in\Omega^p(\stM)\, |
\; -\bdelta\bd\A=0 
\right\}}/{d\Omega^{p-1}_0(\stM)}.
\end{equation}
Since $\bd$ and $\bdelta$ are linear operators on $\Omega^p(\stM)$, we
have that the numerator of the above expression is a vector space.
Furthermore, quotienting by the exact forms is also linear so formally
$\mathcal{M}^p(\stM)$ is a vector space, elements of which are gauge
equivalent classes of Maxwell solutions denoted by $[\A]$.

Next, choose any Cauchy surface $\Sigma\subset\stM$ with $i:\Sigma\rightarrow
\stM$ and define the antisymmetric bilinear form $\sigma$ on $\mathcal{M}^p
(\stM)\times\mathcal{M}^p(\stM)$ by
\begin{equation}
\sigma( [ \A ], [ \B ]) \equiv  \int_\Sigma i^* \left( [\A ]\wedge
*\bd \B - [\B ]\wedge*\bd\A\right),\label{eq:symplectic_product}
\end{equation}
which is by definition gauge invariant.  It is also non-degenerate and
independent of the choice of Cauchy surface (See Dimock's proposition 6).
Thus $\left(\mathcal{M}^p (\stM),\sigma\right)$ is a suitable symplectic
phase space.

On this phase space we also want to consider linear functions which map
$\mathcal{M}^p(\stM)\rightarrow\real$ defined by $[\A]\mapsto\langle [\A] , f\rangle_\stM$
for all $ f\in\Omega^p_0(\stM)$ with $\bdelta f = 0$.  Such functions are related
to the symplectic form in the following sense\dots

\begin{Prop}\label{Prop:symplectic2}
For $[\A]\in\mathcal{M}^p(\stM)$ and  $f\in\Omega^p_0(\stM)$ where $\bdelta f=0$,
we have
\begin{equation}
\langle [\A] , f\rangle_\stM = -\sigma( [\A], [Ef]).
\end{equation}
\end{Prop}
\begin{proof}
Choose any Cauchy surface $\Sigma$, then from Eq.~\ref{eq:INVprop1} above, we have
for all $[\A]$ that are homogeneous solutions of the Maxwell equation
\begin{equation}
\langle [\A], f\rangle_\stM = -\int_\Sigma i^* \left( [\A]\wedge *\bd Ef + \bdelta[\A]\wedge *Ef
-Ef\wedge *\bd[\A] - \bdelta Ef \wedge *[\A] \right)=-\sigma([\A], Ef).
\end{equation}
Also, recall that $Ef$ is a Lorenz solution and thus belongs to some equivalence
class, therefore giving us the result.
\end{proof}
Furthermore, the symplectic form induces a Poisson bracket operation on the
functions over the phase space. (For a detailed description of how this arises, see
Wald \cite{Wald_QFT} and/or Sect.\ 8.1 of AMR \cite{AMR88}.)  For the linear functions
considered above we calculate
\begin{equation}
\left\{ \sigma([\A],[Ef]),\sigma([A],[Ef']) \right\} = \sigma ([Ef],[Ef']).
\end{equation}

\section{Quantization of the generalized Maxwell field}
\label{sec:Quantization}

For electromagnetism in four dimensions, quantization is complicated by gauge
freedom. Even in Minkowski space this presents serious problems: as shown by
Strocchi~\cite{StrocchiI,StrocchiIII} in the Wightman axiomatic approach, the
vector potential cannot exist as an operator-valued distribution if it is to
transform correctly under the Lorentz group or even display commutativity at
spacelike separations in a weak sense.  We have already seen that the same
gauge freedom exists for the generalized Maxwell p-form field, so we are
expecting the same difficulty here.  However several researchers have already
addressed these issues and quantization of a massless free p-form field in
four-dimensional curved spacetime has, to our knowledge, been discussed in
three papers:

The first is by Folacci~\cite{Fola91} who quantizes $p$-form fields in a
``traditional'' manner by adding a gauge-breaking term to the action  which
then necessitates the introduction of Faddeev-Popov ghost fields to remove
spurious degrees of freedom. This is similar in approach to the Gupta--Bleuler
formalism for the free electromagnetic field in four dimensional spacetimes
\cite{Gupta50,Gupta,Crisp_etal}. However, unlike electromagnetism, which
requires only a single ghost field, the generic quantized $p$-form field
suffers from the phenomenon of having ``ghosts for ghosts,'' thus there are
a multiplicity of fields that need to be handled simultaneously \cite{Towns79,
Siegel80}.

The second paper, by Furlani \cite{Furl95}, employs the full Gupta--Bleuler
method of quantization for the electromagnetic field in four-dimensional
static spacetimes with compact Cauchy surfaces.  He constructs a Fock space
and a representation of the field operator $\A$ that satisfies the Klein-Gordon
equation as an operator identity.  This effectively quantizes all four components
of the one-form field.  To remove the two spurious degrees of freedom requires
applying the Lorenz gauge condition as a constraint on the space of states and
imposing a sesquilinear form that is only positive on the ``physical'' Fock
space.  In a later paper \cite{Furl99} Furlani also treats the quantization of
the Proca field in four-dimensional globally hyperbolic spacetimes.  Within
this paper he collects together many of the classical results referenced
in the preceding section.

The third paper, by Dimock \cite{Dimock}, uses a  more elegant
approach to quantize the free electromagnetic field in
four-dimensional spacetimes which does not introduce gauge breaking
terms and ghost fields. He constructs smeared field operators
$\widehat{[\A]}(\J)$ which may be smeared only with co-closed (divergence-free)
test functions, i.e., $\J\in\Omega^p_0(\stM)$ must satisfy $\bdelta \J=0$. These objects
may be interpreted as smeared gauge-equivalence classes of quantum
one-form fields: formally, $\widehat{[\A]}(\J) = \langle \A, \J\rangle_\stM,$ where
$\A$ is a representative of the equivalence class $[\A]$; since
$\bdelta \J=0$, we have $\langle \bd\Lambda,\J\rangle_\stM = \langle
\Lambda,\bdelta \J\rangle_\stM =0$ so this interpretation is indeed gauge
independent.  The resulting operators satisfy the Maxwell equations
in the weak sense and have the correct canonical commutation
relation. We adapt this approach to the generalized Maxwell field
and quantize in the manner found in Dimock \cite{Dimock} and Wald \cite{Wald_QFT}.

\subsection{Quantization via a Fock Space}

To pass from the classical world into the quantum realm requires
replacing our symplectic phase space with a Hilbert space, while
simultaneously promoting functions on the classical phase space to
self-adjoint operators that act on elements of said Hilbert space.
To maintain correspondence with the classical theory, the commutator
of such operators must be $-i$ times their classical Poisson
bracket. Thus we seek operators on a Hilbert space, indexed by
$u\in\mathcal{M}^p(\stM)$ and denoted $\widehat{[\A]}(u)\equiv
\hat\sigma([\A],[u])$, such that
\begin{equation}
\left[ \hat\sigma([\A],[u]) , \hat\sigma([\A],[u']) \right]=-i\sigma([u],[u']).
\end{equation}

We begin with the construction of our Fock space. Our symplectic phase
space is $(\mathcal{M}^p(\stM),\sigma)$ where elements $[\A]$ of
$\mathcal{M}^p(\stM)$
are gauge equivalent classes of real-valued $p$-form solutions to the
Maxwell equation, as defined in the section above. On this space, choose any
positive-definite, symmetric, bilinear map $\mu:\mathcal{M}^p(\stM)
\times\mathcal{M}^p(\stM)\rightarrow\real$ such that for all $[\A]\in \mathcal{M}^p
(\stM)$ we have
\begin{equation}
\mu([\A],[\A]) = \frac{1}{4}\sup_{[\B]\neq0}\frac{\left[\sigma
([\A],[\B])\right]^2}{\mu([\B],[\B])}.
\end{equation}
Many such $\mu$ of this type exist: For each complex structure $J$ on
$\mathcal{M}^p(\stM)$ which is compatible with $\sigma$ in the sense that
$-\sigma([\A],J[\B])$ is a positive-definite inner product gives rise to
such a $\mu$, although this method does not produce all such $\mu$.  For
further discussion on this point see pp. 41-42 of Wald \cite{Wald_QFT}.

We then define the norm $\parallel\cdot\parallel^2 =2\mu(\cdot,\cdot)$ which
is used to form $\mathfrak{m},$ the completion of $\mathcal{M}^p(\stM)$ with
respect to this norm. Next, define the operator $J:\mathfrak{m}\rightarrow\mathfrak{m}$ by
\begin{equation}
\sigma([\A],[\B]) = 2\mu([\A],J[\B])=([\A],J[\B])_\mathfrak{m},
\end{equation}
where $(\; ,\,)_\mathfrak{m}$ defined in the equation above is the inner
product on $\mathfrak{m}$.  As already indicated above, $J$ endows
$\mathfrak{m}$ with a complex structure.  Furthermore, one can prove
straightforwardly that $J$ satisfies $J^* = - J$ and
$J^* J = {\rm id}_\mathfrak{m}$.

The next step is to complexify $\mathfrak{m}$, i.e. $\mathfrak{m}\rightarrow\mathfrak{m}^\complex$,
and extend $\sigma$, $\mu$, and $J$ by complex linearity.  The resulting complex
space, endowed with the complex inner product
\begin{equation}
([\A],[\B])_{\mathfrak{m}^\complex} = 2\mu(\overline{[\A]},[\B])
\end{equation}
for $[\A], [\B] \in \mathfrak{m}^\complex$ is a complex Hilbert space. The
operator $J$ can be diagonalized into $\pm i$ eigenspaces, as $iJ$ is a bounded,
self-adjoint operator on which we can apply the Spectral Theorem. Therefore, we
can decompose $\mathfrak{m}^\complex$ into two orthogonal subspaces based upon the
eigenvalues of $iJ$.  Define $\mathcal{H}\subset\mathfrak{m}^\complex$ to be the
subspace with eigenvalue $+i$ for the operator $J$, which satisfies the three
properties: (i.) The inner product is positive definite over $\mathcal{H}$,
(ii.)  $\mathfrak{m}^\complex$ is equal to the span of $\mathcal{H}$ and its
complex conjugate space $\overline{\mathcal{H}}$, and (iii.) all elements of
$\mathcal{H}$ are orthogonal to all elements of $\overline{\mathcal{H}}$.
We also define the orthogonal projection map $K:\mathfrak{m}^\complex\rightarrow
\mathcal{H}$ with respect to the complex inner product by $K = \frac{1}{2}({\rm id}%
_{\mathfrak{m}^\complex} - i J)$.  Restricting this map to $\mathfrak{m}$ defines
a real linear map $K:\mathfrak{m}\rightarrow\mathcal{H}$, i.e. a map from the
Hilbert space of gauge-equivalent real-valued solutions of the Maxwell equation to
the complex Hilbert space $\mathcal{H}$. For any $[\A_1],[\A_2]\in\mathfrak{m}$
we have
\begin{equation}
\left( K[\A_1],K[A_2] \right)_\mathcal{H} = -i\sigma\left(\overline{K[\A_1]},K[\A_2]\right)
=\mu([\A_1],[\A_2]) -\frac{i}{2}\sigma([\A_1],[\A_2]).
\end{equation}
Finally, the Hilbert space for our quantum field theory is given by the symmetric
Fock space $\mathfrak{F}_s(\mathcal{H})$ over $\mathcal{H}$, i.e.,
\begin{equation}
\mathfrak{F}_s(\mathcal{H}) = \complex \oplus\left[ \bigoplus_{n=1}^\infty
\left( \otimes_s^n \mathcal{H}\right)\right],
\end{equation}
where $\otimes_s^n\mathcal{H}$ represents the $n$-th order symmetric tensor
product over $\mathcal{H}$.

Our next step is to define the appropriate self-adjoint operators on our Fock space.
Let $[f],[g]\in \mathcal{H}$, then for states in $\mathfrak{F}_s(\mathcal{H})$ of finite
particle number, we define the standard annihilation and creation operators,
$\hat{a}(\overline{[f]})$ and $\hat{a}^*([g])$, respectively, where the annihilation
operator is linear in the argument for the complex conjugate space $\overline{\mathcal{H}}$,
while the creation operator is linear in the argument for elements of  $\mathcal{H}$.
(See the appendix of Wald
\cite{Wald_QFT} for more detail.).  On a dense domain of the Fock space, the operators
satisfy the commutation relation
\begin{equation}
\left [ \hat{a}(\overline{[f]}),\hat{a}^*([g])\right] = \left( [f],[g] \right)_\mathcal{H}
\end{equation}
for all $[f],[g]$, with all other commutators vanishing.

From the analysis of the classical wave solutions in the preceding section, we know
that $E$ is a surjective map of all compactly supported test forms $\J\in\Omega^p_0(\stM)$
which are co-closed into an equivalence class in $\mathcal{M}^p(\stM)$, namely  $[E\J]$.
Furthermore, $\mathcal{M}^p(\stM)\subset \mathfrak{m}$, thus, combined with the orthogonal
projection $K$, we have that $K[E\J]\in\mathcal{H}$.  Therefore, we define the smeared quantum
field operator for all co-closed $\J\in\Omega^p_0(\stM)$) by
\begin{equation}
\widehat{[{\A}]}(\J) = \hat{\sigma}([\A],[E\J])= i \hat{a}(\overline{K[E\J]})-
i \hat{a}^*(K[E\J]).
\end{equation}
Note, this is a slight abuse of the notation used earlier where the argument
of $[\A]$ was an element of the phase space.
We now show that such an operator satisfies the generalized Maxwell equation
and canonical commutation relations in the sense of distributions.
\begin{Prop}\label{prop:WE_and CCR}
For $\J\in\Omega^p_0(\stM)$, $\bdelta\J = 0$ we have\\
(a) $\widehat{[{\A}]}(\J)$ satisfies the generalized Maxwell equation in the weak sense,
i.e., $\widehat{[{\A}]}(\bdelta\bd\J)=0$.\\
(b) $\left[\widehat{[{\A}]}(\J),\widehat{[{\A}]}(\J')\right]=-i\langle \J, E\J'\rangle_\stM$.
In particular, if $\mbox{supp }\J, \mbox{supp }\J'$ are spacelike separated then the
commutator is zero.
\end{Prop}
\begin{proof}
(a)  By definition we have $\widehat{[{\A}]}(\bdelta\bd\J) =i\hat{a}(\overline{K%
[E\bdelta\bd\J]}) - i\hat{a}^*(K[E\bdelta\bd\J])$, so we will show
$[E\bdelta\bd\J]=0$.  For any $\theta\in\Omega^p_0(\stM)$, we have
$\Box E\theta = E\Box\theta = 0$, so $[E\bdelta\bd\theta] = -[E\bd\bdelta\theta]
=-[\bd E \bdelta\theta] = 0$, because all exact forms are in the equivalence class
of zero. Since $\J\in\Omega^p_0(\stM)$, we have $[E\bdelta\bd\J]=0$ and
$\widehat{[{\A}]}(\bdelta\bd\J) = i\hat{a}(0)-i\hat{a}^*(0)=0$.  Notice, this part
of the proposition does not require the co-closed condition $\bdelta\J = 0$.\\
(b)  Substituting the definition of the field operator into the commutator yields
\begin{eqnarray}
[\widehat{[{\A}]}(\J),\widehat{[{\A}]}(\J')] &=&
\left [ \hat{a}(\overline{K[E\J]}),\hat{a}^*(K[E\J'])\right]+
\left [ \hat{a}^*(\overline{K[E\J]}),\hat{a}(K[E\J'])\right]\nonumber\\
&=& \left ( K[E\J], K[E\J']\right)_\mathcal{H}-\overline{\left
( K[E\J], K[E\J']\right)}_\mathcal{H}\nonumber\\
&=&-i\sigma([E\J],[E\J']).
\end{eqnarray}
By Proposition~\ref{Prop:symplectic2} we obtain the desired result.
Finally, if $\mbox{supp }\J$ and $\mbox{supp }\J'$ are spacelike separated,
then $\mbox{\supp }\J \cap \mbox{\supp }E\J' = \emptyset$ and the integral
in the inner product vanishes.
\end{proof}

\subsection{Algebraic/Local Quantum Field Theory}

It is well known that different choices of $\mu$ obviously lead to
different constructions of the Fock space $\mathfrak{F}_s(\mathcal{H})$
and hence unitarily inequivalent quantum field theories \cite{Wald_QFT}.
In Minkowski spacetime, Poincar\'{e} invariance picks out a ``preferred''
$\mu$ which leads to a Hilbert space $\mathcal{H}$ of purely positive
frequency solutions to build the Fock space from.  There are also purely
positive frequency solutions in curved stationary spacetimes where the
time translation Killing field generates an isometry similar to that
of Poincar\'{e} invariance in Minkowski spacetime.  In a general curved
spacetime there may be no such isometries, so purely positive frequency
solutions may no exist and the notion of particles becomes somewhat
ambiguous.
This situation led to the development of the algebraic approach to
quantization, also called local quantum field theory. For a general
review of this topic, we recommend the articles by Buchholz \cite{Buch00},
Buchholz and Haag \cite{Bu&Ha00}, and Wald \cite{Wald06}. For a more
thorough discussion see \cite{BrFrVe03,Ho&Wa01,Ho&Wa02, Dimo80}.
Our notation will closely follow that found in Chapter~4 of B\"{a}r et.\ al.\ \cite{BGP07}.

As our final task, we would like to show that our quantized field
theory can be used to generate the Weyl-system commonly used in
algebraic quantum field theory.  The creation and annihilation
operators are unbounded, so in order to work with bounded operators
we introduce the unitary operators on our Fock space
\begin{equation}
W([u]) = \exp\left( i\hat{\sigma}([\A],[u]) \right)
\end{equation}
for all $[u]\in\mathcal{M}^p(\stM)$.  From the definition of the field
operator and its commutation relations, it is relatively straight
forward to show that this map satisfies
\begin{eqnarray}
i.  &&   W([0])= id_{\mathfrak{F}_s(\mathcal{H})},\\
ii. &&   W(-[u]) = W([u])^*,\\
iii.&&   W([u])\cdot W([v]) = e^{-i\sigma([u],[v])/2} W([u]+[v]).
\end{eqnarray}
(The last relation follows from the Baker-Campbell-Hausdorff formula.)
The canonical commutation relation (CCR) algebra $\mathfrak{A}$ is defined
as the $C^*$-algebra generated by $W([E\J])$ for all co-closed
$\J\in\Omega^p_0(\stM)$.
The CCR-algebra $\mathfrak{A}$ together with the map $W$ forms a
Weyl-system for our symplectic phase space $(\mathcal{M}^p(\stM),\sigma)$
which satisfies the Haag-Kastler axioms as generalized by Dimock
\cite{Dimo80}.  The elements of the algebra are interpreted as the
observables related to the quantum field which satisfy the generalized
Maxwell equation. By Theorem~4.2.9 of B\"{a}r {\it et.\ al.}, this
CCR-representation is essentially unique.

Lastly, we would like to indicate that two very different constructions
of the Weyl-system for a symplectic phase space that could be used
are given in B\"{a}r {\it et.\ al.} \cite{BGP07}.  Unfortunately, both
of the Hilbert space representations they construct are not considered
physical because the states (vectors) in the Hilbert space are not
Hadamard, i.e., the two-point function for these states is not related
to a certain $p$-form  Klein-Gordon bisolution of Hadamard form.
In this manuscript, we have given a framework for the rigorous quantization
of the $p$-form field for which the issue of states being Hadamard can
be addressed in due course. In the case of the 0-form field, the Maxwell
equation and the Klein-Gordon equation are the same, so finding Hadamard
states is straightforward. The issue of Hadamard states for the 1-form field
in four-dimensional globally hyperbolic spacetimes can be found in Fewster
and Pfenning \cite{Pfen03}.  We will complete the discussion of Hadamard
states for the general $p$-form field and develop the quantum weak energy
inequality for these states in our next paper.


\section{Conclusions}

In this manuscript we quantized the generalized Maxwell field $\A$ on
globally hyperbolic spacetimes with compact Cauchy surfaces.  We began
by taking the Maxwell equations into the language of exterior differential
calculus.  The resulting field equation \ref{eq:generalized_potential_eq}
could then be carried to any dimension.  Rather remarkably, we found that
minimally coupled scalar field and the electromagnetic field are actually
two examples of a single $p$-form field theory in arbitrary
dimension.  We then discussed fundamental solutions and the Cauchy problem
for the classical $p$-form field theory where we showed that the Cauchy
problem was well posed if we worked in terms of gauge equivalent classes
of solutions.  This was followed by a discussion of the classical, symplectic
phase space consisting of all real valued gauge equivalent Maxwell
solutions $\mathcal{M}^p(\stM)$ and a non-degenerate antisymetric bilinear form
$\sigma([\A],[\B])$ for $\A,\B\in\mathcal{M}^p(\stM)$.  The theory was then
quantized by promoting functions on the phase space to operators that act on
elements of a Hilbert space.  The appropriately selected operators were shown in
Proposition~\ref{prop:WE_and CCR} to satisfy the generalized Maxwell equation
in the weak sense and have the proper canonical commutation relations.
Finally the Weyl-system for our field theory was developed.



\begin{acknowledgments}
\end{acknowledgments}

I would like to thank C.J.\ Fewster and J.C.\ Loftin for numerous illuminating
discussions, and D. Hunt for pointing out missteps in the originally published
version. I would also like to thank C.J.\ Fewster  for
his hospitality and the hospitality of the Department of Mathematics at the
University of York where part of this research was carried out.  This research
was funded by a grant from the US Army Research Office through the USMA Photonics
Research Center.
\vspace{12pt}


\appendix

\section{Cauchy Problem for $\F$}

\begin{Prop}\label{Prop:A_exist_to_give_F}
Let $F_{(0)}\in\Omega_0^p(\Sigma)$ and $F_{(n)}\in\Omega^{p-1}_0(\Sigma)$
with $\bd F_{(0)}=0$ and $\bdelta F_{(n)}=0$ specify Cauchy data for the
field strength $\F\in\Omega^p(\stM)$, with $0<p\le n$, such that
\begin{equation}
\rho_{(0)}\F = F_{(0)}
\qquad\mbox{ and}\qquad
\rho_{(n)}\F = F_{(n)}.
\label{eq:Fcauchy}
\end{equation}
Given this data, there exists a smooth potential $\A\in\Omega^{p-1}(\stM)$
such that $\F = \bd\A$ satisfies the generalized Maxwell equations  $\bd\F=0$
and $\bdelta\F=0$, as well as the conditions \ref{eq:Fcauchy}.
\end{Prop}
\begin{proof}
We know that $\F = \bd\A$ will satisfy the Maxwell equations if
$\bdelta\bd\A = 0$.  To show that such an $\A$ exists we choose as Cauchy
data:
\begin{equation}
\bd A_{(0)} = F_{(0)},\qquad
A_{(\bd)} = F_{(n)},\label{eq:F_Cauchy_Data}
\end{equation}
while $A_{(n)}$ and $A_{(\bdelta)}$ are arbitrary up to $\bd\A_{(\bdelta)}=0$.

The first thing to address is the existence of $A_{(0)}$.  For
non-compact Cauchy surface $\Sigma$ we could restrict to only
those manifolds that are contractible. By the Poincar\'{e}
lemma for contractible manifolds (Theorem~6.4.18 of AMR
\cite{AMR88}) all closed $p$-forms (for $p>0$) are exact.
Alternatively, we could require that the
compactly supported deRham cohomology group $H^p_c(\Sigma)$ for
$p$-forms on the Cauchy surface be of dimension zero, i.e.
$H^p_c(\Sigma)=\{[0]\}$. This is a restriction on the topology
of the Cauchy surface.  If we do have a trivial deRham
cohomology group then all closed $p$-forms $F_{(0)}$ are exact.
Either the contractible or cohomology condition is sufficient
to allow for the existence of a suitable $A_{(0)}$.

From the initial Cauchy data on $\F$
we have
\begin{equation}
\bdelta A_{(\bd)} = \bdelta F_{(n)} = 0.
\end{equation}
Our Cauchy data for $\A$ has the properties necessary to use
Proposition~\ref{thm_Maxwell_solns}. So $\A$ is a
solution to $-\bdelta\bd\A=0$ and therefore $\F=\bd\A$ is a
solution to the generalized Maxwell equations.

Now we show this also reproduces the Cauchy data.  We evaluate
\begin{equation}
\rho_{(0)}\F = \rho_{(0)}\bd\A = \bd\rho_{(0)}\A
=\bd A_{(0)}= F_{(0)}.
\end{equation}
Next we evaluate
\begin{equation}
\rho_{(n)}\F = \rho_{(n)}\bd\A
=\rho_{(\bd)}\A = A_{(\bd)} = F_{(n)}.
\end{equation}
The remaining two pullbacks are trivially zero since
\begin{equation}
\rho_{(\bd)} \F=\rho_{(\bd)} \bd \A = 0
\end{equation}
and
\begin{equation}
\rho_{(\bdelta)} \F=\rho_{(\bdelta)} \bd \A = \rho_{(0)} \bdelta\bd \A=0.
\end{equation}

\end{proof}
\goodbreak





\end{document}